\RequirePackage{amsmath}
\documentclass[a4paper]{llncs}

\usepackage{amsthm}

\usepackage[T1]{fontenc} % for \textsc inside \section
\usepackage{amssymb}
\usepackage{enumerate}
\usepackage{float}
\usepackage{bookmark}
\usepackage{hyperref}
\usepackage{tikz}
\usetikzlibrary{arrows.meta}
\usetikzlibrary{math}

\usetikzlibrary{calc}
\usetikzlibrary{decorations.pathreplacing}
\usetikzlibrary{shapes}
\usetikzlibrary{arrows}
\usetikzlibrary{decorations}
\usetikzlibrary{decorations.pathmorphing}
\usetikzlibrary{patterns}
\usetikzlibrary{positioning}

\let\originalleft\left
\let\originalright\right
\renewcommand{\left}{\mathopen{}\mathclose\bgroup\originalleft}
\renewcommand{\right}{\aftergroup\egroup\originalright}
\newcommand{\ket}[1]{\left| #1 \right\rangle}

\makeatletter
\usetikzlibrary{decorations,backgrounds}
\def\pgfdecoratedcontourdistance{0pt}
\pgfset{
  decoration/contour distance/.code=%
    \pgfmathsetlengthmacro\pgfdecoratedcontourdistance{#1}}
\pgfdeclaredecoration{contour lineto closed}{start}{%
  \state{start}[
    next state=draw,
    width=0pt,
    persistent precomputation=\let\pgf@decorate@firstsegmentangle\pgfdecoratedangle]{%
    \pgfpathmoveto{\pgfpointlineattime{.5}
      {\pgfqpoint{0pt}{\pgfdecoratedcontourdistance}}
      {\pgfqpoint{\pgfdecoratedinputsegmentlength}{\pgfdecoratedcontourdistance}}}%
  }%
  \state{draw}[next state=draw, width=\pgfdecoratedinputsegmentlength]{%
    \ifpgf@decorate@is@closepath@%
      \pgfmathsetmacro\pgfdecoratedangletonextinputsegment{%
        -\pgfdecoratedangle+\pgf@decorate@firstsegmentangle}%
    \fi
    \pgfmathsetlengthmacro\pgf@decoration@contour@shorten{%
      -\pgfdecoratedcontourdistance*cot(-\pgfdecoratedangletonextinputsegment/2+90)}%
    \pgfpathlineto
      {\pgfpoint{\pgfdecoratedinputsegmentlength+\pgf@decoration@contour@shorten}
      {\pgfdecoratedcontourdistance}}%
    \ifpgf@decorate@is@closepath@%
      \pgfpathclose
    \fi
  }%
  \state{final}{}%
}
\makeatother
\tikzset{
  contour/.style={
    decoration={
      name=contour lineto closed,
      contour distance=#1
    },
    decorate}}
    
\hypersetup{
  pdftitle = {},
  pdfauthor = {Andris Ambainis, Kaspars Balodis, Janis Iraids, Krisjanis Prusis, Juris Smotrovs},
  pdfkeywords = {},
  pdfstartview=FitH,
  unicode=true
}

\newcommand{\ex}[3]{\ensuremath{\textsc{Ex}_{#1}^{#2\mid#3}}}
\newcommand{\exnn}{\ex{2m}{m}{m+1}}

\DeclareMathOperator*{\orf}{\textsc{Or}}
\DeclareMathOperator*{\andf}{\textsc{And}}
\DeclareMathOperator{\dyck}{\textsc{Dyck}}
\DeclareMathOperator{\im}{imbal}

\DeclareMathOperator{\dirtwod}{\textsc{Directed-2D-Connectivity}}
\DeclareMathOperator{\undirtwod}{\textsc{Undirected-2D-Connectivity}}
\DeclareMathOperator{\dirdd}{\textsc{Directed-}d\textsc{D-Connectivity}}
\DeclareMathOperator{\undirdd}{\textsc{Undirected-}d\textsc{D-Connectivity}}

\begin{document}

\bookmarksetup{startatroot}

\title{Quantum Lower Bounds for 2D-Grid and Dyck Language\thanks{Supported by QuantERA ERA-NET Cofund in Quantum Technologies
implemented within the European Union's Horizon 2020 Programme (QuantAlgo project)
and ERDF project 1.1.1.5/18/A/020 ``Quantum algorithms: from complexity theory to experiment''.}
}
\author{\texorpdfstring{Andris Ambainis \and Kaspars Balodis \and J\={a}nis Iraids \and \\ Kri\v{s}j\={a}nis Pr\={u}sis \and Juris Smotrovs}{Andris Ambainis, Kaspars Balodis, J\={a}nis Iraids, Kri\v{s}j\={a}nis Pr\={u}sis, Juris Smotrovs}}

\institute{Center for Quantum Computer Science, Faculty of Computing, University of Latvia. }

\theoremstyle{plain}
\newtheorem{cor}{Corollary}

\maketitle

\bookmarksetup{startatroot}

\begin{abstract}
We show quantum lower bounds for two problems. First, we consider the problem of determining 
if a sequence of parentheses is a properly balanced one ({\em a Dyck word}), with a depth of at most $k$. It has been known that, for any $k$, $\tilde{O}(\sqrt{n})$ queries suffice, with a $\tilde{O}$ term depending on $k$. We prove a lower bound of $\Omega(c^k \sqrt{n})$, showing that the complexity of this problem increases exponentially in $k$. This is interesting as a representative example of star-free languages for which a surprising $\tilde{O}(\sqrt{n})$ query quantum algorithm was recently constructed by Aaronson et al. \cite{AGS18}.

Second, we consider connectivity problems on directed/undirected grid in 2 dimensions, if some of the edges of the grid may be missing. By embedding the ``balanced parentheses'' problem into the grid, we show a lower bound of $\Omega(n^{1.5-\epsilon})$ for the directed 2D grid and
$\Omega(n^{2-\epsilon})$ for the undirected 2D grid.
The directed problem is interesting as a black-box model for a class of classical dynamic programming strategies including the one that is usually used for the well-known edit distance problem. 
We also show a generalization of this result to more than 2 dimensions.
\end{abstract}

\pdfbookmark[1]{Introduction}{intro}
\section{Introduction}
\label{s:intro}

We study the quantum query complexity of two problems:

{\bf Quantum complexity of regular languages.} 
Consider the problem of recognizing whether an $n$-bit string belongs to a given regular language. This models a variety of computational tasks that can be described by regular languages.

In the quantum case, the most commonly used model for studying the complexity of various problems is the query model. For this setting, Aaronson, Grier and Schaeffer \cite{AGS18} recently showed that any regular language $L$ has one of three possible quantum query complexities on inputs of length $n$: 
\begin{itemize}
\item
$\Theta(1)$ if the language can be decided by looking at $O(1)$ first or last symbols of the word;
\item
$\tilde{\Theta}(\sqrt{n})$ if the best way to decide $L$ is Grover's search (for example, for the language consisting of all words containing at least one letter a);
\item
$\Theta(n)$ for languages in which we can embed counting modulo some number $p$ which has quantum query complexity $\Theta(n)$.
\end{itemize} 
As shown in \cite{AGS18}, a language being of complexity $\tilde{O}(\sqrt{n})$ (which includes the first two cases above) is equivalent to it being star-free. Star-free languages are defined as the languages which have regular expressions not containing the Kleene star (if it is allowed to use the complement operation). Star-free languages are one of the most commonly studied subclasses of regular languages and there are many equivalent characterizations of them.

The main open problem from \cite{AGS18} is: 
how quickly does the complexity of the 
$\tilde{O}(\sqrt{n})$-query algorithm grow with the complexity of the language? The algorithm given in
\cite{AGS18} involves nested Grover's search, with the number of levels depending on the complexity of $L$. 
Namely, the complexity of this algorithm is 
$O(\sqrt{n} \log^k n)$ where $k$ is the number of states in the syntactic monoid of $L$. ($k$ is at least the number of states $s$ in the smallest automaton recognizing $L$ and can be exponentially larger than $s$.) The question is: can one avoid this exponential dependence of the complexity on $k$?

The simplest concrete example is the language of balanced parentheses of depth bounded by $k$. This language consists of all sequences of opening parentheses ``\texttt{(}'' and closing parentheses ``\texttt{)}'' satisfying the following requirements: 
\begin{itemize}
\item
at any point, the number of closing parentheses never exceeds the number of opening parentheses (i.e., we never try to close a parenthesis that has not been opened);
\item
at the end of the word, the number of closing parentheses is exactly equal to the number of opening parentheses (i.e. all parentheses have been closed);
\item
at any point, the number of opening parentheses exceeds the number of closing parentheses by at most $k$.
\end{itemize}
For example, for $k=2$, \texttt{()()()}, \texttt{(()())}, \texttt{(())(())} are all valid sequences of parentheses whereas \texttt{())(}, \texttt{()((}, and \texttt{((()))} are not.

For this language, the complexity of the algorithm also grows exponentially with $k$ (a naive application of \cite{AGS18} seems to give $O(\sqrt{n} \log^{ck^3} n)$ for some constant $c$ although it is possible that this can be improved to $O(\sqrt{n} \log^{ck} n)$). The question of whether a better dependence on $k$ can be achieved is open. On the lower bound side, \cite{AGS18} shows a lower bound of $\Omega(n)$ if $k=n/2$ but that means that even an algorithm with a complexity of $O(\sqrt{kn})$ could be possible.

{\bf Our results.} We show that an exponential dependence of the complexity on $k$ is unavoidable. Namely, for the balanced parentheses language, we have
\begin{itemize}
\item
there is a constant $c>1$ such that, for all $k\leq \log n$, the quantum query complexity is $\Omega(c^k \sqrt{n})$;
\item
If $k=c\log n$ for an appropriately chosen constant $c$, the quantum query complexity is $\Omega(n^{1-\epsilon})$. 
\end{itemize}

Thus, the exponential dependence on $k$ is unavoidable and distinguishing sequences of balanced parentheses of length $n$ and depth $\log n$ is almost as hard as distinguishing sequences of length $n$ and arbitrary depth. 

Similar lower bounds have recently been independently proven by Buhrman et al. \cite{buhrman2019quantum} and Magniez \cite{Mag19}.

{\bf Finding paths on a grid.}
The second problem that we consider is graph connectivity on subgraphs of the 2D grid. Consider a 2D grid with vertices $(i, j)$,
$i\in \{0, 1, \ldots, n\}, j\in \{0, 1, \ldots, k\}$ and edges from $(i, j)$ to $(i+1, j)$ and $(i, j+1)$. The grid can be either directed (with edges directed in the order of increasing coordinates) or undirected. We are given an unknown subgraph $G$ of the 2D grid and we can perform queries to variables $x_u$ (where $u$ is an edge of the grid) defined by $x_u=1$ 
if $u$ belongs to $G$ and 0 otherwise.
The task is to determine whether $G$ contains a path from $(0, 0)$ to $(n, k)$ (a directed path in the case of the directed grid and an undirected path for the undirected grid).  

Our interest in this problem is driven by the edit distance problem. In the edit distance problem, we are given two strings $x$ and $y$ and have to determine the smallest number of operations (replacing one symbol by another, removing a symbol or inserting a new symbol) with which one can transform $x$ to $y$. If $|x|\leq n, |y|\leq k$,
the edit distance is solvable in time $O(nk)$ by dynamic programming \cite{wagner1974string}. If $n=k$ then, under the strong exponential time hypothesis (SETH), there is no classical algorithm computing edit distance in time $O(n^{2-\epsilon})$ for $\epsilon>0$ \cite{backurs2015edit} and the dynamic programming algorithm is essentially optimal.

However, SETH does not apply to quantum algorithms.
Namely, SETH asserts that there is no algorithm for general instances of SAT that is substantially better than naive search. Quantumly, a simple use of Grover's search gives a quadratic advantage over naive search. This leads to the question: can this quadratic advantage be extended to edit distance (and other problems that have lower bounds based on SETH)?

Since edit distance is quite important in classical algorithms, the question about its quantum complexity has attracted a substantial interest from various researchers. Boroujeni et al. \cite{boroujeni2018approximating} invented a better-than-classical quantum algorithm and an even faster parallel algorithm for approximating the edit distance within a constant factor. The quantum algorithm was later superseded by a better classical algorithm of \cite{C+18}. However, there are no quantum algorithms computing the edit distance exactly (which is the most important case).

The classical algorithm for edit distance can be viewed as consisting of the following steps:
\begin{itemize}
\item
We construct a weighted version of the directed 2D grid (with edge weights 0 and 1) that encodes the edit distance problem for strings $x$ and $y$, with the edit distance being equal to the length of the shortest directed path from $(0, 0)$ to $(n, k)$.
\item
We solve the shortest path problem on this graph and obtain the edit distance.
\end{itemize}
As a first step, we can study the question of whether the shortest path is of length 0 or more than 0. Then, we can view edges of length 0 as present and edges of length 1 as absent. The question ``Is there a path of length of 0?'' then becomes ``Is there a path from $(0, 0)$ to $(n, k)$ in which all edges are present?''. A lower bound for this problem would imply a similar lower bound for the shortest path problem and a quantum algorithm for it may contain ideas that would be useful for a shortest path quantum algorithm.

{\bf Our results.}
We use our lower bound on the balanced parentheses language to show an $\Omega(n^{1.5-\epsilon})$ lower bound for the connectivity problem on the directed 2D grid. This shows a limit on quantum algorithms for finding edit distance through the reduction to shortest paths. More generally, for an $n\times k$ grid ($n>k$), our proof gives a lower bound of
$\Omega((\sqrt{n}k)^{1-\epsilon})$.

The trivial upper bound is $O(nk)$ queries, since there are $O(nk)$ variables. There is no nontrivial quantum algorithm, except for the case when $k$ is very small. Then, the connectivity problem can be solved with $O(\sqrt{n} \log^{k-1} n)$ quantum queries \cite{Kle17} \footnote{Aaronson et al. \cite{AGS18} also give a bound of $O(\sqrt{n} \log^{m-1} n)$ but in this case $m$ is the rank of the syntactic monoid which can be exponentially larger than $k$.} but this bound becomes trivial already for $k=\Omega(\frac{\log n}{\log \log n})$.

For the undirected 2D grid, we show a lower bound of $\Omega((nk)^{1-\epsilon})$, whenever $k\geq \log n$. Thus, the naive algorithm is almost optimal in this case. We also extend both of these results to higher dimensions, obtaining a lower bound of 
$\Omega((n_1 n_2 \ldots n_d)^{1-\epsilon})$ for an undirected $n_1 \times n_2 \times \ldots \times n_d$ grid in $d$ dimensions and a lower bound of 
$\Omega(n^{(d+1)/2-\epsilon})$ for a directed $n\times n \times \ldots \times n$ grid in $d$ dimensions.

In a recent work, an $\Omega(n^{1.5})$ lower bound for edit distance was shown by Buhrman et al. \cite{buhrman2019quantum}, 
assuming a quantum version of Strong Exponential Time hypothesis (QSETH). One of steps in this lower bound 
is an $\Omega(n^{1.5})$ query lower bound for a different path problem on a 2D grid. Then, QSETH is invoked to prove
that no quantum algorithm can be faster than the best algorithm for this shortest path problem. 
Our main lower bound result is of similar nature as the query part of \cite{buhrman2019quantum} but neither of the two
results follow directly one from another, because of different shortest path problems that are used.

\section{Definitions}
\label{s:defs}

For a word $x\in\Sigma^*$ and a symbol $a\in\Sigma$ by $|x|_a$ we denote the number of occurrences of $a$ in $x$.

For two (possibly partial) Boolean functions $g: G \rightarrow \{0,1\}$, where $G \subseteq \{0,1\}^n$, and $h:H\rightarrow \{0,1\}$, where $H\subseteq \{0,1\}^m$, we define the composed function $g\circ h: D\rightarrow \{0,1\}$, with $D\subseteq \{0,1\}^{nm}$, as
\[
\left(g\circ h\right) (x) = g\left(h(x_1,\dots,x_m), \dots, h(x_{(n-1)m+1},\dots, x_{nm})\right).
\]
Given a Boolean function $f$ and a nonnegative integer $d$, we define $f^d$ recursively as $f$ iterated $d$ times: $f^d=f\circ f^{d-1}$ with $f^1=f$.

{\bf Quantum query model.}
    Quantum query model is a generalization of the decision tree model of classical computation that is commonly used to lower bound the amount of time required by a computation. Let $f:D\rightarrow \{0,1\},D\subseteq \{0,1\}^n$ be an $n$ variable function we wish to compute on an input $x\in D$. We have an oracle access to the input $x$ --- it is realized by a specific unitary transformation usually defined as $\ket{i}\ket{z}\ket{w}\rightarrow \ket{i}\ket{z+x_i\pmod{2}}\ket{w}$ where $\ket{i}$ register indicates the index of the variable we are querying, $\ket{z}$ is the output register, and $\ket{w}$ is some auxiliary work-space. An algorithm in the query model consists of alternating applications of arbitrary unitaries independent of the input and the query unitary, and a measurement in the end. The smallest number of queries for an algorithm that outputs $f(x)$ with probability $\geq \frac{2}{3}$ on all $x$ is called the quantum query complexity of function $f$ and denoted by $Q(f)$.
    
    Let a symmetric matrix $\Gamma$ be called an adversary matrix for $f$ if the rows and columns of $\Gamma$ are indexed by inputs $x\in D$ and $\Gamma_{xy}=0$ if $f(x)=f(y)$. Let $\Gamma^{(i)}$ be a similarly sized matrix such that $\Gamma^{(i)}_{xy}=\begin{cases}\Gamma_{xy}&\text{ if }x_i\neq y_i\\ 0&\text{ otherwise}\end{cases}$. Then let 
    \[Adv^{\pm}(f)=\max_{\Gamma\text{ - an adversary matrix for }f}{\frac{\|\Gamma\|}{\max_i{\|\Gamma^{(i)}\|}}}\]
    be called the adversary bound and let
    \[Adv(f)=\max_{\substack{\Gamma\text{ - an adversary matrix for }f\\ \Gamma \text{ - nonnegative}}}{\frac{\|\Gamma\|}{\max_i{\|\Gamma^{(i)}\|}}}\]
    be called the positive adversary bound.
    
    The following facts will be relevant for us:
    \begin{itemize}
        \item $Adv(f)\leq Adv^\pm(f)$;
        \item $Q(f)=\Theta(Adv^{\pm}(f))$ \cite{Reichardt11};
        \item $Adv^{\pm}$ composes exactly even for partial Boolean functions $f$ and $g$, meaning, $Adv^\pm(f\circ g)=Adv^\pm(f)\cdot Adv^\pm(g)$ \cite[Lemma~6]{kimmel2012quantum}
    \end{itemize}
    
{\bf Reductions.}
    We will say that a Boolean function $f$ is reducible to $g$ and denote it by $f \leqslant g$ if there exists an algorithm that given an oracle $O_x$ for an input of $f$ transforms it into an oracle $O_y$ for $g$ using at most $O(1)$ calls of oracle $O_x$ such that $f(x)$ can be computed from $g(y)$. Therefore, from $f \leqslant g$ we conclude that $Q(f)\leq Q(g)$ because one can compute $f(x)$ using the algorithm for $g(y)$ and the reduction algorithm that maps $x$ to $y$.

{\bf Dyck languages of bounded depth.}
Let $\Sigma$ be an alphabet consisting of two symbols: \texttt{(} and \texttt{)}.
The Dyck language $L$ consists of all $x\in \Sigma^*$ that represent a correct sequence of opening and closing parentheses.
We consider languages $L_k$ consisting of all words $x\in L$ where the number of opening parentheses that are not closed yet never exceeds $k$.

The language $L_k$ corresponds to a query problem $\dyck_{k, n}(x_1, ..., x_n)$ where $x_1, \ldots, x_n \in \{0, 1\}$ describe a word of length $n$
in the natural way: the $i^{\rm th}$ symbol of $x$ is \texttt{(} if $x_i=0$ and \texttt{)} if $x_i=1$. $\dyck_{k, n}(x)=1$ iff the word $x$ belongs to $L_k$.

{\bf Connectivity on a directed 2D grid.} Let $G_{n, k}$ be a directed version of an $n\times k$ grid in two dimensions, with vertices $(i, j), i\in [n], j\in [k]$
and directed edges from $(i, j)$ to $(i+1, j)$ (if $i<n$) and from $(i, j)$ to $(i, j+1)$ (if $j<k$).
If $G$ is a subgraph of $G_{n, k}$, we can describe it by variables $x_e$ corresponding to edges $e$ of $G_{n, k}$: $x_e=1$ 
if the edge $e$ belongs to $G$ and $x_e=0$ otherwise. We consider a problem $\dirtwod$ in which one has to determine if $G$ contains a 
path from $(0, 0)$ to $(n, k)$: 
$$\dirtwod_{n, k}(x_1, \ldots, x_m)=1$$ (where $m$ is the number of edges in $G_{n, k}$)
iff such a path exists. 

{\bf Connectivity on an undirected 2D grid.} Let $G'_{n, k}$ be an undirected $n\times k$ grid and let $G$ be a subgraph of $G_{n, k}$. We describe $G$ by 
variables $x_e$ in a similar way and define $\undirtwod_{n, k}(x_1, \ldots, x_m)=1$ iff $G$ contains a path from $(0, 0)$ to $(n, k)$.

We also consider $d$ dimensional versions of these two problems, on $n\times n \times \ldots n$ grids (with the grid being of the same length in all the dimensions). 
In the directed version ($\dirdd$), we have a subgraph $G$ of a directed grid (with edges directed in the directions from $(0, \ldots, 0)$ to
$(n, \ldots, n)$) and 
$$\dirdd(x_1, \ldots, x_m)=1$$
iff $G$ contains a directed path from $(0, \ldots, 0)$ to
$(n, \ldots, n)$.
The undirected version is defined similarly, with an undirected grid instead of the directed one.

\section{Lower bounds for Dyck languages}
\label{s:dyck}

\begin{theorem}
\label{t:dycklowerbound-power}
There exist constants $c_{1},c_{2}>0$ such that $Q\left(\dyck_{c_{1}\ell m,c_{2}\left(2m\right)^{\ell}}\right)=\Omega\left(m^{\ell}\right)$.
\end{theorem}

\begin{proof}
We will use the partial Boolean function $\ex{m}{a}{b}=\begin{cases}
1, & \text{if \ensuremath{\left|x\right|_{0}=a}}\\
0, & \text{if \ensuremath{\left|x\right|_{0}=b.} }
\end{cases}$ 

We prove the theorem by a reduction $\left(\exnn\right) ^{\ell}\leqslant\dyck_{c_{1}\ell m,c_{2}\left(2m\right)^{\ell}}$.
It is known that $Adv^\pm\left(\exnn\right) \geq Adv\left(\exnn\right)>m$ \cite[Theorem~5.4]{ambainis2002quantum}.
The Adversary bound composes even for partial Boolean functions \cite[Lemma~1]{kimmel2012quantum}, therefore $Q\left( \left(\exnn\right)^{\ell}\right)=\Omega\left(m^{\ell}\right)$.
Via the reduction the same bound applies to $\dyck_{c_{1}\ell m,c_{2}\left(2m\right)^{\ell}}$.

Before we describe the reduction in detail, we sketch the main idea. %Consider mapping an input of $\ex{2m}{m}{m+1}$ to an input of $\dyck_{k,n}$. 
Let $\im(x)=|x|_0-|x|_1$. Note that \[\ex{2m}{m}{m+1}(x)=0 \iff \im(x)=2\] \[\ex{2m}{m}{m+1}(x)=1 \iff \im(x)=0\] whereas 
\begin{equation*}\dyck_{k,n}(x)=1\iff \begin{aligned}\max_{p\text{ -- prefix of }x}{\im(p)}\leq k \wedge \\ \min_{p\text{ -- prefix of }x}{\im(p)}\geq 0 \wedge \\ \im(x)=0.\end{aligned}\end{equation*} If we could make sure that the minimum and maximum constraints are satisfied, $\dyck_{k,n}$ could be used to compute $\exnn$. To ensure the minimum constraint, we map each $0$ to $00$ and $1$ to $01$. However, this increases $\im(x)$ by $2m$ which can be fixed by appending $1^{2m}$ at the end. Importantly, the resulting sequence $x'$ has $\im(x')=\im(x)$. The first constraint (maximum over prefixes) can be fulfilled by having a sufficiently large $k$; $k=2m+3$ would suffice here. The same idea can be applied iteratively to $\ex{2m}{m}{m+1}$ where the inputs, which could now be the results of functions $\left(\ex{2m}{m}{m+1}\right)^{\ell-1}=x_i$, have been recursively mapped to sequences $x_i'$ with $\im(x_i')=\begin{cases}2 \text{ if }x_i=0\\0\text{ if }x_i=1\end{cases}$.

The reduction formally is as follows.

We call a string $B\in\left\{ 0,1\right\} ^{w}$ of even length
a $\left(w,h\right)$-sized block with width $w$ and height $h$
iff for any prefix $x$ of $B$: $0\leq \im(x)\leq h$
and either $\im(B)=0$ or $\im(B)=2$.

We establish a correspondence between inputs to $\left(\exnn\right)^\ell$ that satisfy the promise and  
$\left(w,h\right)$-sized blocks $B$
for appropriately chosen $w, h$,
so that $\left(\exnn\right)^\ell=1$
iff $\im(B)=0$.

For $l=0$ (the input bits), we have $0$ corresponding to a $(2,2)$-sized block of $00$ and $1$ to a $(2,2)$-sized block of $01$.

For $l>0$,
let us have input bits $x=(x_1, x_2, \ldots, x_{2m})$ of $\exnn$ satisfying the input promise. Assume that the bits (that could be equal to values of $\left(\exnn\right)^{\ell-1}$) correspond to $(w,h)$-sized blocks $B_1, B_2, \ldots, B_{2m}$. Define the sequence $B'=B_1B_2\ldots B_{2m}1^{2m}$. Then it is easy to verify the following claims:
\begin{enumerate}[1)]
    \item $B'$ is a $(2m(w+1),2(m+1)+h)$-sized block;
    \item The output bit of $\exnn(x)$ \emph{corresponds} to $B'$ because
    \[\im(B')=\sum_{i=1}^{2m}{\im(B_i)}+\im(1^{2m})=\begin{cases}2\text{ if }\exnn(x)=0\\ 0\text{ if }\exnn(x)=1\end{cases}.\]
\end{enumerate}

For $l=0$, the inputs correspond to $(2,2)$-sized blocks. Each level adds $2(m+1)$ to the height of the blocks reaching $2+2\ell(m+1)=O(m\ell)$. The width of blocks reaches $O((2m)^\ell)$.

Since for all $(w,h)$-sized blocks $B$: $\dyck_{h,w}(B)=1\iff \im(B)=0$ one can solve the $\left(\exnn\right)^\ell$ problem by running $\dyck_{h,w}$ on the \emph{corresponding} block.

See Figure \ref{fig:ex-dyck-reduction}.

\begin{figure}[H]
\begin{centering}
\begin{tikzpicture}[scale=0.7,transform shape]
\tikzstyle{invis} = [outer sep=0,inner sep=0,minimum size=0]
\tikzstyle{vt} = [very thick]
\draw [step=1cm, very thin,lightgray] (1,0) grid (18.5,4.5) node (v2) {};

\node[invis,label=left:{$0$}] at (1,0) {};
\node[invis,label=left:{$m$}] at (1,1) {};
\foreach \x in {2,...,4}
{
	\node[invis,label=left:{$\x m$}] at (1,\x) {};
}

\draw [vt](16,2) -- (18,0);
\draw [vt](16,2.1) -- (18,0.1);

\draw [decorate, decoration={brace, mirror, amplitude=6}](1,0) -- (16,0) node[invis,midway,yshift=-0.4cm]{$2m\cdot 6m$};
\draw [decorate, decoration={brace, mirror, amplitude=6}](16,0) -- (18,0) node[invis,midway,yshift=-0.4cm]{$2m$};

\draw[vt,fill=white](1,0) rectangle (2.9,2.1) node[midway] {$\exnn$};

\draw[vt,fill=white](3,0) rectangle (4.9,2.2) node[midway] {$\exnn$};

\draw[vt,fill=white](5,0) rectangle (6.9,2.3) node[midway] {$\exnn$};

\node at (8.5,1.5) {\Huge{$\dots$}};g

\draw[vt,fill=white](10,1.7) rectangle (11.9,4.1) node[midway] {$\exnn$};
\draw[vt,fill=white](12,1.8) rectangle (13.9,4.1) node[midway] {$\exnn$};
\draw[vt,fill=white](14,1.9) rectangle (15.9,4.1) node[midway] {$\exnn$};

\draw [vt](2.9,0) -- +(0.1,0);
\draw [vt](2.9,0.1) -- +(0.1,0);

\draw [vt](4.9,0) -- +(0.1,0);
\draw [vt](4.9,0.1) -- +(0.1,0);
\draw [vt](4.9,0.2) -- +(0.1,0);

\draw [vt](6.9,0) -- +(0.1,0);
\draw [vt](6.9,0.1) -- +(0.1,0);
\draw [vt](6.9,0.2) -- +(0.1,0);
\draw [vt](6.9,0.3) -- +(0.1,0);

\draw [vt](15.9,2) -- +(0.1,0);
\draw [vt](15.9,2.1) -- +(0.1,0);

\draw [vt](13.9,1.9) -- +(0.1,0);
\draw [vt](13.9,2) -- +(0.1,0);
\draw [vt](13.9,2.1) -- +(0.1,0);

\draw [vt](11.9,1.8) -- +(0.1,0);
\draw [vt](11.9,1.9) -- +(0.1,0);
\draw [vt](11.9,2) -- +(0.1,0);
\draw [vt](11.9,2.1) -- +(0.1,0);

\draw [vt](9.9,1.7) -- +(0.1,0);
\draw [vt](9.9,1.8) -- +(0.1,0);
\draw [vt](9.9,1.9) -- +(0.1,0);
\draw [vt](9.9,2) -- +(0.1,0);
\draw [vt](9.9,2.1) -- +(0.1,0);

\draw[vt,red, dashed](1,0) rectangle (18,4.1);

\end{tikzpicture}
\par\end{centering}
\caption{\label{fig:ex-dyck-reduction}The reduction $\exnn\circ\exnn\protect\leqslant\dyck_{4m+6,12m^{2}+2m}$.
The line of the graph follows the input word along the \emph{x}-axis and shows
the number of yet-unclosed parenthesis along the \emph{y}-axis  (i.e., a zoomed-out version of Figure \ref{f:pyramids}).
The input word $B_1 B_2 \dots B_{2m} 1^{2m}$ corresponds to the outer function $\exnn$ with $B_j$
being a block corresponding to the output of an inner $\exnn$.
The ticks at the starts and ends of blocks depict that if the line enters the block at height $i$, it exits at height $i$ or $i+2$. In the block the line never goes below $0$ or above $h+i$.
The red dashed part then forms a new block $B'$.
By replacing the blocks $B_j$ with blocks $B'$ we can further iterate $\exnn$ to get the reduction $\exnn\circ\left(\exnn\right)^{\ell-1}\protect\leqslant\dyck_{O\left(\ell m\right),O\left(\left(2 m \right)^{\ell}\right)}$.}
\end{figure}
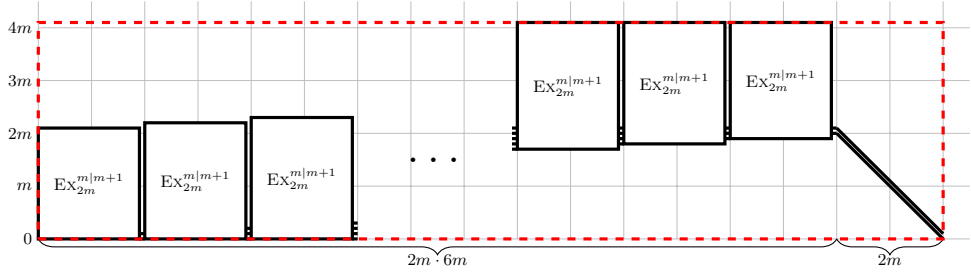

\end{proof}

\begin{theorem}
\label{t:dycklowerbound}
For any $\epsilon > 0$, there exists $c>0$ such that \[Q\left(\dyck_{c\log n,n}\right)=\Omega\left(n^{1-\epsilon}\right).\]
\end{theorem}

\begin{proof}
For any $\epsilon>0$, there exists an $m$ such that $Adv^\pm\left(\ex{2m}{m}{m+1}\right)\geq\left(2m\right)^{1-\epsilon}$.
Without loss of generality we may assume that $(2m)^\ell=n$. Using the reduction from Theorem \ref{t:dycklowerbound-power} with $\ell=\log_{2m}n$ we obtain a sequence with length $c_2\left(2m\right)^{\ell}=c_2 n$
and height $c_1 m\ell=\Theta\left(\log n\right)$. The query complexity
is at least $\left(\left(2m\right)^{1-\epsilon}\right)^{\ell}=\left(\left(2m\right)^{\ell}\right)^{1-\epsilon}=n^{1-\epsilon}$.
Therefore $Q\left(\dyck_{c \log n,n}\right)=\Omega\left(n^{1-\epsilon}\right)$.
\end{proof}

For constant depths the following bound can be derived:
\begin{theorem}
There exists a constant $c_1>0$ such that
    \[Q(\dyck_{c_1\ell,n})=\Omega(2^{\frac{\ell}{2}}\sqrt{n}).\]
\end{theorem}
\begin{proof}
    Let $m=4$ in the Theorem \ref{t:dycklowerbound-power}. Thus there exist constants $c_{1},c_{2}>0$ such that $Q\left(\dyck_{c_{1}\ell,c_{2}8^{\ell}}\right)=\Omega\left(4^{\ell}\right)$. Consider the function $\andf_\frac{n}{c_{2}8^{\ell}} \circ \dyck_{c_{1}\ell,c_{2}8^{\ell}}$ where we have promise that $\andf_k$ has as an input either $k$ or $k-1$ ones. The query complexity of this function is $\Omega\left(\sqrt{\frac{n}{c_{2}8^{\ell}}}4^{\ell}\right)=\Omega(2^\frac{\ell}{2}\sqrt{n})$. The computation of the composition $\andf_\frac{n}{c_{2}8^{\ell}} \circ \dyck_{c_{1}\ell,c_{2}8^{\ell}}$ can be straightforwardly reduced to $\dyck_{c_{1}\ell,n}$ by a simple concatenation of $\dyck_{c_{1}\ell,c_{2}8^{\ell}}$ instances.
\end{proof}

\section{Lower bounds for \textsc{st-Connectivity} in grids}
\label{s:connectivity}

\subsection{Lower bounds for $\dirtwod_{n, k}$}
\label{ss:lbdirected}

\begin{theorem}
\label{t:dirconlowerbound}
Let $n\geq k$. Then for any $\epsilon > 0$
\[Q(\dirtwod_{n, k})=\Omega\left((\sqrt{n}k)^{1-\epsilon}\right).\]
\end{theorem}
In particular, if we have a square grid then
\begin{cor}
For a square grid for any $\epsilon > 0$ \[Q(\dirtwod_{n, n})=\Omega\left(n^{1.5-\epsilon}\right).\]
\end{cor}

\begin{proof}[Proof of Theorem \ref{t:dirconlowerbound}]
For any sequence $w$ of $m$ opening and closing parentheses it is possible to plot the changes of depth, i.e., the number of opening parentheses minus the number of closing parentheses, for all prefixes of the sequence, see Figure \ref{f:pyramids}.
\begin{figure}
    \centering
    \begin{tikzpicture}[scale=0.88]
    \tikzmath{
    int \x,\y,\newx,\newy,\depth;
    \y=0;
    \x=0;
    \depth=4;
    for \delta in {1,1,-1,-1,1,1,1,-1,-1,1,-1,-1}%
    {
        \newx=\x+1;
        if (\delta > 0) then
        {
            \paren = "(";
        }   
        else
        {
            \paren = ")";
        };
        \newy=\y+\delta;
        {
            \path (\x,0) -- (\newx,0) node [midway, below,yshift=-20pt] {\LARGE \texttt{\paren}};
            \draw[-Stealth,very thick] (\x,\y) -- (\newx,\newy);
        };
        \y=\newy;
        \x=\newx;
    };
    {
        \draw[->] (-0.5,0) -- node[above,pos=1] {$x$} (\newx+0.5, 0);
        \draw[->] (0,-0.5) -- node[left,pos=1] {$y$} (0, \depth+0.5);
        \draw[dashed] (-0.5,\depth) -- node[above, pos=0.9] {$y=d=4$} (\newx+0.5, \depth);
        \draw (-0.5, 0.3) node {$(0,0)$};
        \path [preaction={contour=8pt,rounded corners=10,draw}] (0,0) -- (\depth,\depth) -- (\newx-\depth, \depth) -- (\newx, 0) -- cycle;
    };
    };
    \end{tikzpicture}    
    \caption{Representation of the Dyck word ``\texttt{(())((())())}''}
    \label{f:pyramids}
\end{figure}
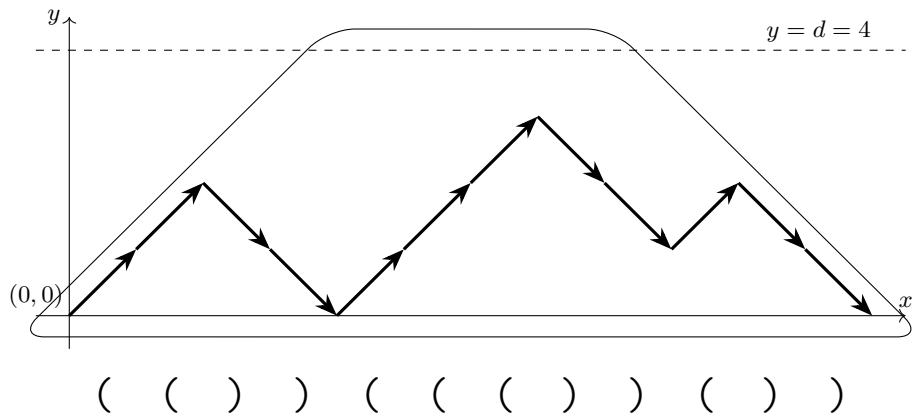
We can connect neighboring points by vectors $(1,1)$ and $(1,-1)$ corresponding to opening and closing parentheses respectively. Clearly $w\in L_d$ if and only if the path starting at the origin $(0,0)$ ends at $(m,0)$ and never crosses $y=0$ and $y=d$. Consequently a path corresponding to $w\in L_d$ always remains within the trapezoid bounded by $y=0$, $y=d$, $y=x$, $y=-x+m$. This suggests a way of mapping $\dyck_{d,m}$ to the $\dirtwod_{n, k}$ problem:
\begin{enumerate}
    \item An opening parenthesis in position $i$ corresponds to a ``column'' of upwards sloping available edges $(i-1,l)\rightarrow (i,l+1)$ for all $l \in \{0,1, \ldots, d-1\}$ such that $i-1+l$ is even.
    A closing parenthesis in position $i$ corresponds to downwards sloping available edges $(i-1,l)\rightarrow (i,l-1)$ for all $l \in \{1, \ldots, d\}$ such that $i-1+l$ is even. See Figure \ref{f:varmap}.
    \begin{figure}[H]
    \centering
    \begin{tikzpicture}[scale=0.8]
    \path (0,0) -- (0,8) node [midway, xshift=-20pt] {\LARGE \texttt{(}$\implies$};
    \path (4,0) -- (4,8) node [midway, xshift=-20pt] {\LARGE \texttt{)}$\implies$};
        \tikzmath{
            for \y in {0,...,3}%
            {
                {
                    \draw[-Stealth] (0,2*\y) -- ++(1,1);
                    \draw[-Stealth,dotted] (0,2*\y+2) -- ++(1,-1);
                    \draw[-Stealth,dotted] (4,2*\y) -- ++(1,1);
                    \draw[-Stealth] (4,2*\y+2) -- ++(1,-1);
                };
            };
        };
    \end{tikzpicture}    
    \caption{$\dyck_{d,m}$ to $\dirtwod$ variable mapping}
    \label{f:varmap}
    \end{figure}
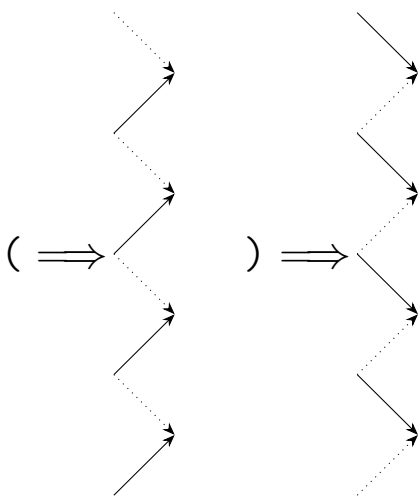    

    \item The edges outside the trapezoid adjacent to the trapezoid are forbidden (see Figure \ref{f:pyramids2}), i.e., it is sufficient to ``insulate'' the trapezoid by a single layer of forbidden edges. The only exception are the edges adjacent to the $(0,0)$ and $(m,0)$ vertex as those will be used in the construction (step 4).
    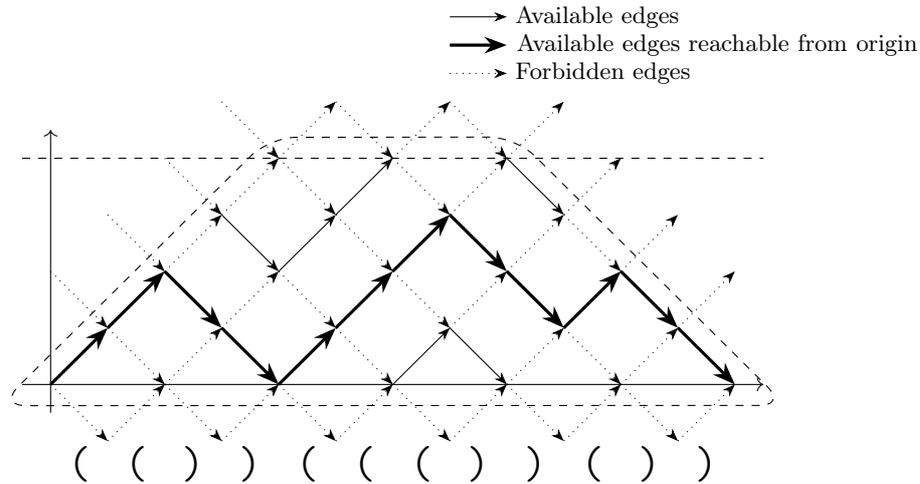
\begin{figure}[H]
    \centering
    \begin{tikzpicture}[scale=0.75]
    \tikzmath{
    int \x,\y,\newx,\newy,\depth;
    \y=0;
    \x=0;
    \depth=4;
    for \delta in {1,1,-1,-1,1,1,1,-1,-1,1,-1,-1}%
    {
        \newx=\x+1;
        if (\delta > 0) then
        {
            \paren = "(";
        }   
        else
        {
            \paren = ")";
        };
        \newy=\y+\delta;
        {
            \path (\x,0) -- (\newx,0) node [midway, below,yshift=-20pt] {\LARGE \texttt{\paren}};
            \draw[-Stealth,very thick] (\x,\y) -- (\newx,\newy);
        };
        \y=\newy;
        \x=\newx;
    };
    {
        \draw[-Stealth,dotted] (0,0) -- (1,-1);
        \draw[-Stealth,dotted] (1,-1) -- (2,0);
        \draw[-Stealth,dotted] (2,0) -- (3,-1);
        \draw[-Stealth,dotted] (3,-1) -- (4,0);
        \draw[-Stealth,dotted] (4,0) -- (5,-1);
        \draw[-Stealth,dotted] (5,-1) -- (6,0);
        \draw[-Stealth,dotted] (6,0) -- (7,-1);
        \draw[-Stealth,dotted] (7,-1) -- (8,0);
        \draw[-Stealth,dotted] (8,0) -- (9,-1);
        \draw[-Stealth,dotted] (9,-1) -- (10,0);
        \draw[-Stealth,dotted] (10,0) -- (11,-1);
        \draw[-Stealth,dotted] (11,-1) -- (12,0);

        \draw[-Stealth,dotted] (0,2) -- (1,1);
        \draw[-Stealth,dotted] (1,3) -- (2,2);
        \draw[-Stealth,dotted] (2,4) -- (3,3);
        \draw[-Stealth,dotted] (3,5) -- (4,4);
        \draw[-Stealth,dotted] (4,4) -- (5,5);
        \draw[-Stealth,dotted] (5,5) -- (6,4);
        \draw[-Stealth,dotted] (6,4) -- (7,5);
        \draw[-Stealth,dotted] (7,5) -- (8,4);
        \draw[-Stealth,dotted] (8,4) -- (9,5);
        \draw[-Stealth,dotted] (9,3) -- (10,4);
        \draw[-Stealth,dotted] (10,2) -- (11,3);
        \draw[-Stealth,dotted] (11,1) -- (12,2);
        
        \draw[-Stealth] (3,3) -- (4,2);
        \draw[-Stealth] (4,2) -- (5,3);
        \draw[-Stealth] (5,3) -- (6,4);
        \draw[-Stealth] (6,0) -- (7,1);
        \draw[-Stealth] (7,1) -- (8,0);
        \draw[-Stealth] (8,4) -- (9,3);

        \draw[-Stealth,dotted] (1,1) -- (2,0);
        \draw[-Stealth,dotted] (2,0) -- (3,1);
        \draw[-Stealth,dotted] (2,2) -- (3,3);
        \draw[-Stealth,dotted] (3,1) -- (4,2);
        \draw[-Stealth,dotted] (3,3) -- (4,4);
        \draw[-Stealth,dotted] (4,2) -- (5,1);
        \draw[-Stealth,dotted] (4,4) -- (5,3);
        \draw[-Stealth,dotted] (5,1) -- (6,0);
        \draw[-Stealth,dotted] (5,3) -- (6,2);
        \draw[-Stealth,dotted] (6,2) -- (7,1);
        \draw[-Stealth,dotted] (6,4) -- (7,3);
        \draw[-Stealth,dotted] (7,1) -- (8,2);
        \draw[-Stealth,dotted] (7,3) -- (8,4);
        \draw[-Stealth,dotted] (8,0) -- (9,1);
        \draw[-Stealth,dotted] (8,2) -- (9,3);
        \draw[-Stealth,dotted] (9,1) -- (10,0);
        \draw[-Stealth,dotted] (9,3) -- (10,2);
        \draw[-Stealth,dotted] (10,0) -- (11,1);

        \draw[->] (-0.5,0) -- (\newx+0.5, 0);
        \draw[->] (0,-0.5) -- (0, \depth+0.5);
        \draw[dashed] (-0.5,\depth) -- (\newx+0.5, \depth);
        \path [preaction={contour=8pt,rounded corners=10,draw,dashed}] (0,0) -- (\depth,\depth) -- (\newx-\depth, \depth) -- (\newx, 0) -- cycle;
        \draw[-Stealth] (7,6.5) -- node[pos=1,right,align=left] {Available edges} ++(1,0);
        \draw[-Stealth,very thick] (7,6) -- node[pos=1,right,align=left] {Available edges reachable from origin} ++(1,0);
        \draw[-Stealth,dotted] (7,5.5) -- node[pos=1,right,align=left] {Forbidden edges} ++(1,0);
    };
    };
    \end{tikzpicture}    
    \caption{Mapping of a complete input corresponding to Dyck word ``\texttt{(())((())())}'' to $\dirtwod$}
    \label{f:pyramids2}
\end{figure}
    \item Rotate the trapezoid by 45 degrees counterclockwise. This isolated trapezoid can be embedded in a directed grid and its starting and ending vertices are connected by a path if and only if the corresponding input word is valid. 
    \item Finally we can lay multiple independent trapezoids side by side and connect them in parallel forming an $\orf_t$ of $\dyck_{d,m}$ instances; see Figure \ref{f:trapezoids}.
    \begin{figure}[H]
    \centering
    \begin{tikzpicture}[scale=0.4]
        \draw [step=1,dotted] (0,0) grid (30,10);

        \foreach \x in {0,5,...,21} 
        {
        	\draw[thick,fill=gray!30,dashed] (\x,0)--++(0,4)--++(5,5)--++(4,0) coordinate (exit\x) -- cycle;
        	\draw[very thick,->] (exit\x)--++(0, 1);
        }

        \draw[very thick,->] (0,0)--(20,0);
        \draw[very thick,->] (9,10)--(30,10);
    
    \end{tikzpicture}    
    \caption{Reduction $\orf_t\circ \dyck \leqslant \dirtwod$}
    \label{f:trapezoids}
    \end{figure}
\end{enumerate}
This concludes the reduction $\orf_t\circ \dyck_{d,m} \leqslant \dirtwod_{n,k}$, where $n=(d+1)t+\frac{m}{2}-d-1$ and $k=\frac{m}{2}+1$. By the well known composition result of Reichardt \cite{Reichardt11} we know that $Q(\orf_t\circ \dyck_{d,m})=\Theta(Q(\orf_t)\cdot Q(\dyck_{d,m}))$.
All that remains is to pick suitable $t$, $d$ and $m$ for the proof to be complete. Let $k$ be the vertical dimension of the grid and $k\leq n$. Then we take $m=\Theta(k)$, $d=\log{m}$ and $t=\frac{n}{d}$.

\end{proof}

\subsection{Lower bounds for $\undirtwod_{n, k}$}
Even though it is possible to use the construction from Section \ref{ss:lbdirected} to give a lower bound of $\Omega\left((\sqrt{n}k)^{1-\epsilon}\right)$ for the undirected case because the paths for each instance of $\dyck$ never bifurcate or merge, this lower bound can be further improved to a nearly tight estimate.

\begin{theorem}
\label{t:undirconlowerbound}
Let $n\geq k$ and $k=\Omega(\log{n})$. Then for any $\epsilon>0$
\[Q(\undirtwod_{n, k})=\Omega\left((nk)^{1-\epsilon}\right).\]
\end{theorem}

\begin{proof}
    We start off by representing an input as a path in a trapezoid, see Figure \ref{f:pyramids2}. But now instead of connecting multiple instances of $\dyck$ in parallel we will embed one long instance by folding it when it hits the boundary of the graph. To implement a fold we will use simple gadgets depicted in Figure \ref{f:foldgadget}.
    \begin{figure}[H]
    \centering
    \begin{tikzpicture}[scale=0.27]
        \draw [step=1,dotted] (0,0) grid (44,20);

        	\draw[thick,fill=gray!30,dashed] (0,0)--++(0,4)--++(16,16)--++(2,-2) -- cycle;
        	\draw[thick,fill=gray!30,dashed] (19,17)--++(2,-2)--++(-12,-12)--++(-2,2) -- cycle;
        	\draw[thick,fill=gray!30,dashed] (12,0)--++(-2,2)--++(18,18)--++(2,-2) -- cycle;
        	\draw[thick,fill=gray!30,dashed] (31,17)--++(2,-2)--++(-12,-12)--++(-2,2) -- cycle;
        	\draw[thick,fill=gray!30,dashed] (24,0)--++(-2,2)--++(18,18)--++(4,0) -- cycle;
            
            \draw[thick] (16,20)--++(5,0)--++(0,-5);
            \draw[thick] (17,19)--++(3,0)--++(0,-3);
            \draw[thick] (18,18)--++(1,0)--++(0,-1);
            \draw[thick] (28,20)--++(5,0)--++(0,-5);
            \draw[thick] (29,19)--++(3,0)--++(0,-3);
            \draw[thick] (30,18)--++(1,0)--++(0,-1);
            \draw[thick] (12,0)--++(-5,0)--++(0,5);
            \draw[thick] (11,1)--++(-3,0)--++(0,3);
            \draw[thick] (10,2)--++(-1,0)--++(0,1);
            \draw[thick] (24,0)--++(-5,0)--++(0,5);
            \draw[thick] (23,1)--++(-3,0)--++(0,3);
            \draw[thick] (22,2)--++(-1,0)--++(0,1);

    \end{tikzpicture}    
    \caption{Folding of a long $\dyck$ instance in an undirected grid}
    \label{f:foldgadget}
    \end{figure}
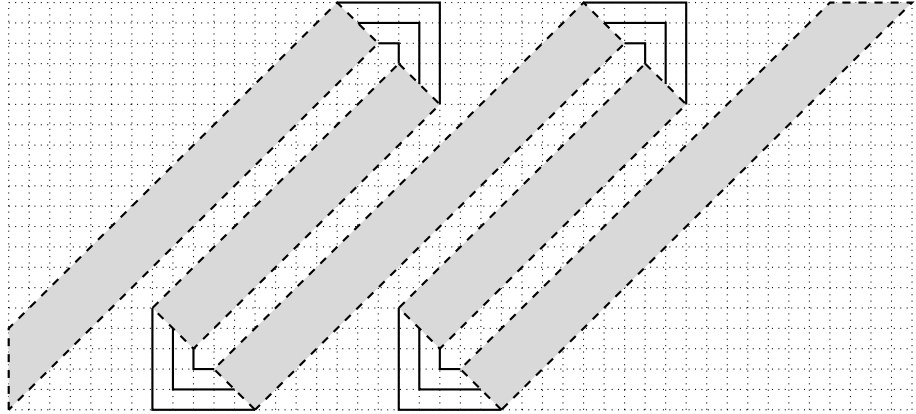
This way a $\dyck$ instance of length $m$ and depth $\log{m}$ can be embedded in an $n \times k$ grid such that $\frac{nk}{\log{m}}=\Theta(m)$. Using Theorem \ref{t:dycklowerbound} we conclude that solving $\undirtwod_{n,k}$ requires at least $\Omega\left((nk)^{1-\epsilon}\right)$ quantum queries.
\end{proof}

\subsection{Lower bounds for $d$-dimensional grids}
For undirected $d$-dimensional grids we give a tight bound on the number of queries required to solve connectivity.
\begin{theorem}
    For any $\epsilon>0$, for undirected $d$-dimensional grids of size $n_1\times n_2 \times \ldots \times n_d$ that are not ``almost-one-dimensional'', i.e., there exists $i\in [d]$ such that $\frac{\prod_{j=1}^d{n_j}}{n_i}=\Omega(\log{n_i})$:
    \[Q(\undirdd_{n_1, n_2,\ldots, n_d}) = \Omega((n_1 \cdot n_2 \cdot \ldots \cdot n_d)^{1-\epsilon}).\]
\end{theorem}
\begin{proof}
    This follows from the 2D case by using the fact that a $d$-dimensional grid of size $n_1\times n_2\times \ldots \times n_{d-1}\times n_d$ contains as a subgraph a $(d-1)$-dimensional grid of size $n_1 \times n_2 \times \ldots \times n_{d-2} \times n_{d-1}n_d$. One way to see this is to consider a bijective mapping of the vertices $(x_1, \ldots, x_{d-1}, x_d)$ to $(x_1, \ldots, x_{d-2},x_dn_{d-1}+x_{d-1})$ if $x_d$ is even and to $(x_1, \ldots, x_{d-2},x_dn_{d-1}+n_{d-1}-1-x_{d-1})$ if $x_d$ is odd. It is a bijection because $x_d$ and $x_{d-1}$ can be recovered from $x_dn_{d-1}+n_{d-1}-1-x_{d-1}$ by computing the quotient and remainder on division by $n_{d-1}$. One can view this procedure as ``folding'' where we take layers (vertices corresponding to some $x_d=l$) and fold them into the $(d-1)$-st dimension alternating the direction of the layers depending on the parity of the layer $l$.
\end{proof}

For directed $d$-dimensional grids we can only slightly improve over the trivial lower bound of $n^{\frac{d}{2}}$.
\begin{theorem}\label{t:dirdlowerbound}
    For directed $d$-dimensional grids of size $n_1\times n_2 \times \ldots \times n_d$ such that $n_1\leq n_2 \leq \ldots \leq n_d$ and $\epsilon>0$:
    \[Q(\dirdd_{n_1,n_2,\ldots,n_d}) = \Omega((n_{d-1}\prod_{i=1}^d{n_i})^{\frac{1}{2}-\epsilon}).\]
\end{theorem}

\begin{cor}
For directed $d$-dimensional grids of size $n\times n \times \ldots \times n$ and $\epsilon>0$:
\[Q(\dirdd_{n,n,\ldots,n}) = \Omega(n^{\frac{d+1}{2}-\epsilon}).\]
\end{cor}
\begin{proof}[Proof of Theorem \ref{t:dirdlowerbound}]
    For each $I\in [n_1]\times [n_2]\times \ldots \times [n_{d-2}]$ we  take take a $2$-dimensional hard instance $G_I$ of $\dirtwod_{n_{d-1},n_d}$ having query complexity $\Omega(n_{d-1}^{1-\epsilon}n_d^{
    \frac{1}{2}-\epsilon})$. We then connect them in parallel like so:
    \begin{itemize}
        \item Make available the entire $(d-2)$-dimensional subgrid from $(1, 1, \ldots, 1, 1, 1)$ to $(n_1, n_2, \ldots, n_{d-2}, 1, 1)$ and similarly the subgrid from $(1, 1, \ldots, 1, n_{d-1}, n_d)$ to $(n_1, n_2, \ldots, n_{d-2}, n_{d-1}, n_d)$;
        \item For each $I\in [n_1]\times [n_2]\times \ldots \times [n_{d-2}]$ embed the instance $G_I$ in the subgrid $(I, 1, 1)$ to $(I, n_{d-1}, n_d)$;
        \item Forbid all other edges.
    \end{itemize}
    This construction computes $\orf_{\prod_{i=1}^{d-2}{n_i}}\circ \dirtwod_{n_{d-1},n_d}$ whose complexity is at least $\Omega(\sqrt{\prod_{i=1}^{d-2}{n_i}}n_{d-1}^{1-\epsilon}n_d^{\frac{1}{2}-\epsilon})=\Omega((n_{d-1}\prod_{i=1}^d{n_i})^{\frac{1}{2}-\epsilon})$.
\end{proof}

\pdfbookmark[1]{Conclusion}{concl}
\section{Conclusion}
\label{s:concl}

We have shown quantum lower bounds on the complexity of two problems in the query model:
\begin{itemize}
\item
recognizing Dyck languages of bounded depth (i.e. determining whether a sequence of parentheses is a balanced sequence of depth at most $k$);   
\item
determining connectivity on grids of dimension 2 and more where some edges may be missing (and we have query access to the information whether edges are present).
\end{itemize}
The first bound shows that the complexity increases exponentially in $k$, as $\Omega(\sqrt{n} c^{k})$. This provides a lower bound counterpart to the recent result of Aaronson et al. \cite{AGS18} who constructed an $\tilde{O}(\sqrt{n})$ quantum  algorithm for all star-free languages (which includes the Dyck languages of bounded depth) where $\tilde{O}$ term has exponential dependence on the complexity of the language (measured by the rank of its syntactic monoid).

The two bounds are not completely matching, as they involve related but different measures of complexity in the exponent (possible depth of the   sequence of parentheses vs. the rank of the syntactic monoid). Also, the lower bound has $c^k$ dependence while the upper bound has $\log^k n$ dependence. However, our result shows that some type of exponential dependence on a natural measure of complexity of the language is unavoidable.

We then used our results on Dyck languages to show the lower bounds on finding paths in 2D grids, by embedding sequences of parentheses into the grid so that existence of a path between two vertices corresponds to a sequence of parentheses. Having a lower bound with an exponential dependence on $k$ was essential for obtaining good lower bounds for the connectivity problems.

The resulting lower bounds for connectivity on 2D grids are $\Omega(n^{1.5-\epsilon})$ for the directed grid and $\Omega(n^{2-\epsilon})$ for the undirected grid. These results generalize to $\Omega(n^{(k+1)/2-\epsilon})$ for the directed grid
and $\Omega(n^{k-\epsilon})$ for the undirected grid in $k$ dimensions. The best upper bound is the trivial
$O(n^k)$ query quantum algorithm in all cases\footnote{We can show that the result of Ronagh \cite{ronagh2019quantum} on speeding up dynamic programming implies a $O(n^{5+k/2})$ query algorithm for the directed connectivity in $k$ dimensions. However, we have identified a possible problem in the proof of \cite{ronagh2019quantum} and are not aware of ways to fix it.}.

The lower bound for the directed 2D grid is important as it implies an $\Omega(n^{1.5-\epsilon})$ lower bound on the most natural approach towards a quantum algorithm for edit distance (by reducing it to connectivity on 2D grid). More generally, we think that the connectivity problems introduced in this paper are a natural new class of problems that is worth studying in the quantum algorithms context. 

Some directions for future work are:
\begin{enumerate}
    \item {\bf Better algorithm/lower bound for the directed 2D grid?} 
    Can we find an $o(n^2)$ query quantum algorithm or improve the lower bound in this paper? A nontrivial quantum algorithm would be particularly interesting, as it may imply a quantum algorithm for edit distance.  
    \item {\bf Quantum algorithms for directed connectivity?} More generally, can we come up with better quantum algorithms for directed connectivity? The span program method used by Belovs and Reichardt \cite{BelovsR12} for the undirected connectivity does not work in the directed case. As a result, the quantum algorithms for directed connectivity are typically based on Grover's search in various forms, from simply speeding up depth-first/breadth-first search to more sophisticated approaches \cite{A+19}. Developing other methods for directed connectivity would be very interesting.
    \item {\bf Quantum speedups for dynamic programming.} Dynamic programming is a widely used algorithmic method for classical algorithms and it would be very interesting to speed it up quantumly. This has been the motivating question for both the connectivity problem on the directed 2D grid studied in this paper and a similar problem for the Boolean hypercube in \cite{A+19} (where faster-than-classical algorithms were discovered for both the directed connectivity on hypercube and the problems that motivated the study of it). There are many more dynamic programming algorithms and exploring quantum speedups of them would be quite interesting.  
\end{enumerate}

\bibliographystyle{plain}

\phantomsection
\addcontentsline{toc}{chapter}{References}
\bibliography{biblio}

\end{document}